%% file: 2020-acc-LPControlPlanning.tex
\documentclass[letterpaper, 10 pt, conference]{ieeeconf}  

\IEEEoverridecommandlockouts                              

\input{preamble/common}
\graphicspath{{./figures/}}

\title{\LARGE \bf
  Robust Path Planning and Control For Polygonal Environments via Linear Programming
}

\author{Mahroo Bahreinian$^{1}$, Erfan Aasi$^{2}$ and Roberto Tron$^{3}$
  \thanks{This work was supported by ONR MURI N00014-19-1-2571 ``Neuro-Autonomy: Neuroscience-Inspired Perception, Navigation, and Spatial Awareness''}
  \thanks{$^{1}$Mahroo Bahreinian is with Division of Systems Engineering at Boston University, Boston, MA, 02215 USA. Email:
    {\tt\small mahroobh@bu.edu}}%
  \thanks{$^{2}$Erfan Aasi is with Department of Mechanical Engineering at Boston University, Boston, MA, 02215 USA. Email:
    {\tt\small eaasi@bu.edu}}%
  \thanks{$^{3}$Roberto Tron is with Faculty of Department of Mechanical Engineering at Boston University, Boston, MA, 02215 USA. Email:
    {\tt\small tron@bu.edu}}%
}

\begin{document}
\maketitle
\thispagestyle{empty}
\pagestyle{empty}

\begin{abstract}

  We propose a novel approach for navigating in polygonal environments by synthesizing controllers that take as input relative displacement measurements with respect to a set of landmarks.
  Our algorithm is based on solving a sequence of robust min-max Linear Programming problems on the elements of a cell decomposition of the environment.
  The optimization problems are formulated using linear Control Lyapunov Function (CLF) and Control Barrier Function (CBF) constraints, to provide stability and safety guarantees, respectively.
  The inner maximization problem ensures that these constraints are met by all the points in each cell, while the outer minimization problem balances the different constraints in a robust way.
  We show that the min-max optimization problems can be solved efficiently by transforming it into regular linear programming via the dualization of the inner maximization problem.
  We test our algorithm to agents with first and second order integrator dynamics, although our approach is in principle applicable to any system with piecewise linear dynamics.
  Through our theoretical results and simulations, we show that the resulting controllers: are optimal (with respect to the criterion used in the formulation), are applicable to linear systems of any order, are robust to changes to the start location (since they do not rely on a single nominal path), and to significant deformations of the environment.
\end{abstract}
\section{INTRODUCTION}
Path planning is a major research area in the context of mobile robots, as it deals with the problem of finding a path from an initial state toward a goal state while considering collision avoidance. Traditional path planning methods focus on finding \emph{single nominal paths} in a given \emph{known map}, and the majority of them makes the implicit assumption that the agent possesses a lower-level \emph{state feedback} controller for following such nominal path in the face of external disturbances and imperfect models. Biological system do not rely on the same restrictive assumptions; for instance, consider a person navigating in an unfamiliar room: despite the fact that the person does not have a precise blueprint of the floor, and does not know its precise location, they can reliably and robustly navigate toward a desired door, from any location in the room. While this ability in biological systems is the result of complex and not fully understood mechanisms, in this paper we aim to narrow the gap between planning algorithms and biological systems by 
synthesizing \emph{output-feedback controllers} that are robust to \emph{imprecise map knowledge}. By synthesizing controllers instead of specific paths, we more tightly integrate the high-level path planning and the low-level regulation tasks, allowing the agent to cope with disturbances without replanning; moreover, the focus on the controllers allows us to plan by directly using measurements (outputs) available to the agent, instead of assuming full state knowledge; finally, since the controllers depend on the environment indirectly (through measurements that are taken online), we empirically show that such controllers are robust to (often very significant) changes in the map. In order to pursue strong theoretical guarantees, in this paper we assume agents with controllable linear dynamics, and environments that admit a polygonal convex cell decomposition (e.g., via Delaunay triangulations \cite{fortune1992voronoi} or trapezoidal decompositions \cite{latombe2012robot}). Methods to address these limitations are planned as part of our future work (see also the Conclusions section).

\myparagraph{Previous work}
Existing works on path planning can be roughly classified into two categories: combinatorial path planning methods, and sample-based path planning methods \cite{comparative}.
Some of the path planning methods consider a continuous model for the environment and therefore provide a continuous path, such as potential fields \cite{khatib1986real}, \cite{krogh1984generalized} and navigation functions \cite{rimon1992exact}, while the other group solves the planning problem by abstracting the environment to a finite representation and find a discrete path, such as probabilistic roadmaps \cite{kavraki1996probabilistic} and cell decomposition methods \cite{lingelbach2004path}.

One of the well known combinatorial path planning algorithms is cell decomposition where a complex environment is decomposed into a set of cells, avoiding obstacles by planning straight paths in individual cells; for each individual step, traditional methods use midpoints \cite{lavalle2006planning, schurmann2009computational,choset2005principles}, while more recent solutions aim to optimize path length \cite{kloetzer2015optimizing}. Our work can be seen as a descendant of previous work that handles the cell decomposition vis-\'a-vis the continuous dynamic through a hybrid system perspective by synthesizing a state-feedback controller for each cell. Initial work proposed potential-based controllers \cite{conner2003composition}, while others characterize the theoretical conditions \cite{habets2006reachability}
and closed-form solutions \cite{belta2005discrete} for linear affine controllers. Although the latter approaches were extended to nonlinear systems in \cite{girard2008motion} and to uncertain maps \cite{yan2008mobile} (using intelligent re-planning), they all assume that each cell in the decomposition is a \emph{simplex} (a polytope in $\real{d}$ with $d+1$ vertices, e.g., a 2-D triangle). In contrast, our method can handle arbitrary convex polytopes, and design \emph{output}-feedback controllers (instead of state-feedback).
 
Sampling-based planning algorithms, such as rapidly exploring random tree (\texttt{RRT}), have become popular in last few years due to their good practical performance, and their probabilistic compleness \cite{lavalle2006planning,lavalle2001randomized,karaman2011sampling}. For trajectory planning that takes into account non-trivial dynamical systems of the robot, kinodynamic \texttt{RRT} \cite{lavalle2001randomized, lavalle2006planning} and closed-loop \texttt{RRT} (\texttt{CL-RRT}, \cite{kuwata2008motion}) and \texttt{CL-RRT\#} grow the tree  by sampling control inputs and then propagating forward the nonlinear dynamics (with the optional use of stabilizing controllers and tree rewiring to approach optimality). Further in this line of work, there has been a relatively smaller amount of works on algorithms that focus on producing controllers as opposed to simple reference trajectories. 


The \texttt{safeRRT} algorithm \cite{positiveInvariant,weiss2017motion} generates a closed-loop trajectory  from initial state to desired goal by expanding a tree of local state-feedback controllers to maximize the volume of corresponding positive invariant sets while satisfy the input and output constraints. 
Based on the same idea and following the \texttt{RRT} approach, the \texttt{LQR-tree} algorithm \cite{tedrake2009lqr} creates a tree by sampling over state space and stabilizes the tree with an linear quadratic regulator (LQR) feedback. With respect to the present paper, the common traits among all these works is the use of a full state feedback (as opposed to output feedback), although they do not require prior knowledge of convex cell decomposition of the environment.

Finally, our work builds upon real-time synthesis of point-wise controls that trade off safety and stability for nonlinear input-affine systems through a Quadratic Program (QP) formulation \cite{ames2014control,hsu2015control}. To the best of our knowledge, our paper is the first to use similar conditions for synthesizing controls over entire convex regions rather than single points.

\myparagraph{Proposed approach and contributions}
In this work, we propose a novel approach to synthesize a set of output-feedback controllers on a convex cell decomposition of a polygonal environment via Linear Programming (LP). We define constraints in terms of a Control Lyapunov Function (CLF) and Control Barrier Functions (CBF) to ensure, respectively, stability and safety (collision avoidance) throughout all the states in a cell, while automatically balancing the two aspects to maximize robustness. Our formulation results in a linear min-max optimization problem, which is solved by converting it to a LP form.

With respect to previous work:
\begin{enumerate*}
\item We allow a cell to be any generic convex polytope (instead of a simplex).
\item We consider output feedback based on any affine function of the state (under the natural assumption that the overall dynamics is controllable), although, for the sake of presenting a concrete application, we focus on controls using measurements of the relative position of the agent with respect to landmarks in the environment.
\item We apply the CLF-CBF to the new framework of control synthesis.
\end{enumerate*}
We believe that our solution can be extended to sample-based methods and non-linear systems, although these are beyond the scope of the current paper (see the Conclusions section for details).

\section{NOTATION AND PRELIMINARIES}\label{notation}

In this section we review CLF and CBF constraints in the context of our application on agents with linear dynamics and a convex cell decomposition of the environment.

\subsection{System dynamics}
We start by considering a control-affine dynamical system\footnote{The CLF-CBF concepts are applicable to input-affine systems, but in this work we assume linear time-invariant systems, and affine barrier functions.}
\begin{equation}\label{sys1}
  \dot{x}=Ax+Bu,
\end{equation}
where $x \in \cX$ denotes the state, $u \in \cU$ the system input, and $A\in\real{n_x\times n_x}$, $B\in\real{n_x\times n_u}$ define the linear dynamics, and $\cX\subset \real {n_x}$, $\cU\subset\real{n_u}$ denotes limits on the states, and actuators, respectively.
 We assume $\cX_{\textrm{dyn}}$ and $\cU$ are polytopic, 
\begin{align}\label{cbf_ff}
  \cX_{\textrm{dyn}}=\{x\mid A_{x,\textrm{dyn}}x\leq b_{x,\textrm{dyn}}\},&&
  \cU=\{u\mid A_uu\leq b_u\},
\end{align}
where $A_\textrm{x,dyn}\!\in\! \real {s_d\times n_x}$, $A_{u} \!\in\! \real {s_u\times n_u}$, $b_{\textrm{x,dyn}} \!\in\! \real {n_x}$, $b_{u}\!\in\! \real{n_u}$, and that $0\in \cX_{\textrm{dyn}}$. $s_d$ and $s_u$ are the number of dynamic constraints and controller constraints respectively. 

\newcommand{\xpos}{x_{\textrm{pos}}}
\newcommand{\xdyn}{x_{\textrm{dyn}}}
\newcommand{\Ppos}{P_{\textrm{pos}}}
\newcommand{\Pdyn}{P_{\textrm{dyn}}}
Since the system \eqref{sys1} can be higher-order, but the environment constrains only positions, we give the following.
\begin{definition}\label{xposdyn}
  We assume that a subset of the state $x$ in \eqref{sys1} represents the position $\xpos=\Ppos x$ of the agent in the world, while $\xdyn=\Pdyn x$ represents the rest of the state (e.g., velocities in a second order system), where $\Ppos\in \real{d\times n_x}$ and $\Pdyn \in \real{(n_x-d)\times n_x}$ are orthogonal projection matrices.
\end{definition}
\begin{remark}
  
  In this section, we only define constraints for the dynamic part of the states, $\xdyn$, i.e., $\Ppos A_{\xdyn}^T=0$. The constraints on $\xpos$ will be derived from the environment.
\end{remark}
\subsection{High relative degree functions and transverse dynamics}
Given a function $h$ of the state of the dynamical system \eqref{sys1}, the following notions characterize the relation between the derivatives along the system's trajectories and the inputs $u$ of the system. Note that we assume that $h$ is sufficiently smooth so that all the necessary derivatives are well defined.
\begin{definition}
  The Lie derivative of a differentiable function $h$ for the dynamics \eqref{sys1} with respect to the vector field $Ax$ is defined as $\cL_{Ax}h(x)=\frac{\partial h(x(t))}{\partial x}\transpose Ax$. 
  The Lie derivative of order $r$ is denoted as $\cL_{Ax}^r$, and is recursively defined by $\cL_{Ax}^{r}h(x)=\cL_{Ax}(\cL_{Ax}^{r-1}h(x))$, with $\cL_{Ax}^1h(x)=\cL_{Ax}h(x)$ \cite{zcbf2}.
\end{definition}
\begin{definition}
  A function $h(x)$ is said to have relative degree $r$ with respect to the dynamics \eqref{sys1} if $\cL_B\cL_{Ax}^{i}h(x)= 0$ for $0 \leq i \leq r-1$ and $\cL_B\cL_{Ax}^{r}h(x)\neq 0$; equivalently, it is the minimum order of the time derivative of the system, $h^{r}(x)$, that explicitly depends on the inputs $u$. The Lie derivative of $h(x)$ with relative degree $r$ for dynamics \eqref{sys1} is defined as
  \begin{equation}
    h^r(x)=\cL_{Ax}^rh(x)+\cL_B\cL_{Ax}^{r-1}h(x)u
  \end{equation}
\end{definition}

\begin{definition}\label{def:transversal}
Given a function $h(x)$ with relative degree $r$ for the dynamics \eqref{sys1}, we define the \emph{transversal state} 
\begin{equation}\label{xi1}
  \xi_h(x)=\bmat{h(x)\\\dot{h}(x)\\\vdots\\h^{r-1}(x)}=\bmat{h(x)\\\cL_{Ax}h(x)\\\vdots\\\cL_{Ax}^{r-1}h(x)},
\end{equation}
and the \emph{transversal dynamics}
\begin{equation}\label{xi2}
\begin{aligned}
  &\dot{\xi}_h(x)=F\xi_h(x)+G\mu_h,\\
  &h(x)=C\xi_h
  \end{aligned}
\end{equation}
where $F\in\real{r\times r}$ , $G\in\real r$ and $C\in \real{1\times r}$ are defined as
\begin{align}\label{xi3}
  F=\begin{bmatrix}
    0 & 1 & \hdots & 0\\
    \vdots & \vdots & \ddots & \vdots \\
    0 & 0 & 0&1\\
    0 & 0 & 0& 0\\
  \end{bmatrix},&& G=\begin{bmatrix}0\\ \vdots \\0\\1
  \end{bmatrix} &&  C=[1 &0\;\hdots\;0],
\end{align}
and the virtual control input 
$\mu_h=h^r$
is a function of the actual input $u$.
\end{definition}
We use these concepts below to define higher-order CBFs and CLFs.
\begin{remark}
  When $h(x)$ is an affine function (see Sec.~\ref{sec_cbf}), all Lie derivatives are linear functions of $x$, and the system \eqref{xi3} can be interpreted as \eqref{sys1} in observable canonical form.
\end{remark}

\subsection{Safety Constraints by Control Barrier Function}
Suppose we have a sufficiently smooth function $h(x):\real {n} \rightarrow \real{}$ which defines a safe set $\cC$ such that
\begin{equation}\label{set_c}
  \begin{aligned}
    \cC&=\{x\in \real{n}|\;h(x)\geq0\},\\
    \partial \cC&=\{x\in \real{n}|\;h(x)=0\},\\
    {Int}(\cC)&=\{x\in \real{n}|\;h(x)>0\}.
  \end{aligned}
\end{equation}
We say that the set $\cC$ is \emph{forward invariant} (also said \emph{positive invariant} \cite{positiveInvariant}) if  $x(t_0) \in \cC$ implies $x(t)\in \cC$, for all $t\geq t_0$ where $x(t)$ is well defined~\cite{zcbf1}.


\begin{definition}[ECBF, \cite{nguyen2016exponential}]\label{def_L}
  Consider the control system \eqref{sys1}, and a continuously differentiable function $h(x)$ with relative degree $r\geq 0$ defining a set $\cC$ as in \eqref{set_c}. The function $h(x)$ is an \textit{Exponential Control Barrier Function} (ECBF) if there exist $c_h \in \real {r}$ and control inputs $u\in \cU$ such that
  \begin{equation}\label{ecbf_con}
    \cL_{Ax}^{r}h(x)+\cL_B\cL_{Ax}^{r-1}h(x)u+c_h\transpose\xi_h(x)\geq 0,\forall x \in Int(\cC).
  \end{equation}
\end{definition}
\begin{proposition}\label{prop:ECBF}
  Given an ECBF $h(x)$ and control inputs $u$ from Definition~\ref{def_L}, if $c_h$ stabilizes the transversal dynamics, i.e., the closed-loop matrix $A_h-B_hc_h\transpose$ is stable, then
  \begin{enumerate}
  \item $\cL_{A}^{j}h(x)\geq -p_{1}\cL_{A}^{j-1}h(x)$ for $1\leq j \leq r$ ,
  $p_{1}\geq 0 $
  \item the set $\cC$ is forward invariant.
  \end{enumerate}
\end{proposition}
\begin{proof}
  From \cite[equation (41) and Theorem 2]{nguyen2016exponential}, for the family of outputs $y_i:\real{n}\rightarrow\real{}$ for $i=1,\hdots,r$ we have 
  \begin{equation}
      y_i=\dot{y}_{i-1}+p_iy_{i-1}
  \end{equation}
  where $p_i\in \real{+}$ for $i=1,\hdots,r$ is a pole location of $\mu_h = -c_h^T\xi_h$, $C_i=\{x\in \real{n}|y_i\geq 0\}$ ,and $y_0=\cL_{A}^{0}h(x)$ so 
  \begin{equation}
      y_1={\cL_{A}^{0}h(x)}+p_1\cL_{A}^{0}h(x)\geq 0 
  \end{equation}
  which implies $\cL_{A}^{j}h(x)\geq -p_1\cL_{A}^{j-1}h(x)$.
  
  For proof of claim 2 see  \cite[Theorem 1]{nguyen2016exponential}.
\end{proof}
 Note that if $r=1$, $h(x)$ is also a special case of a \textit{Zeroing Control Barrier Function} (ZCBF,\cite{zcbf1,zcbf2}). 

\subsection{Stability Constraints by Control Lyapunov Function }\label{CLF}
In this section we present an analogous definition extending CLFs \cite{ames2014rapidly} to higher-order relative degrees.
\begin{definition}\label{def:ECLF}
  Consider the control system \eqref{sys1}, and a continuously differentiable function $V(x)$ defined over a set $\cX$ with $V(x)\geq 0$ and relative degree $r\geq 0$. The function $V(x)$ is a \textit{Exponential Control Lyapunov Function} (ECLF) 
  if there exists $c_V\in\real{r}$ and control inputs $u\in \cU$ such that
  \begin{equation}\label{L2}
    \cL_A^{r}V(x)+\cL_B\cL_A^{r-1}V(x)u+c\transpose_V\xi_V(x)\leq 0,\forall x \in \cX.
  \end{equation}
\end{definition}
For $r=1$, we recover the definition of Exponentially Stabilizing CLFs (ES-CLFs, \cite{ames2014rapidly}).
It is possible to use the ECLF to design controllers that exponentially stabilize the original dynamics \eqref{sys1}, as shown by the following:

\begin{proposition}\label{prop:ECLF}
  Given an ECLF $V(x)$ and controls $u$ from Definition~\ref{def:ECLF}, if $\cX$ is a forward-invariant set, and $c\transpose_V$ stabilizes the transversal dynamics, i.e., the matrix $F-Gc\transpose_V$ is stable, then:
  \begin{enumerate}
  \item\label{it:der} $\cL_{A}^{j}V(x)\leq -q_1\cL_{A}^{j-1}V(x)$ for $1\leq j \leq r$ ,
  $q_{1}\geq 0 $
  \item\label{it:Vlimit} $\lim_{t\to\infty} V(x(t))=0$ with exponential convergence;
  \item\label{it:xlimit} If $V(x)$ in addition satisfies
  \begin{equation}\label{L1}
    \alpha_1(\norm{x})\leq V(x)\leq\alpha_2(\norm{x})
  \end{equation}
  where $\alpha_1,\alpha_2$ are class-$\cK$ functions, then $\lim_{t\to\infty}x=0$ with exponential convergence.
\end{enumerate}
\end{proposition}

\begin{proof}
The proof mirrors a simplified version of the ideas in \cite{nguyen2016exponential}. Setting the virtual input $\mu_V$ in the transversal
dynamics to $\mu_V=-c\transpose_V\xi_V$,
   we have that $\lim_{t\to\infty} \xi_V=0$ with exponential convergence (since it is an LTI system and $c_V$ contains stabilizing feedback gains). We can apply $\mu_V\leq -c\transpose_V\xi_V$, then
  \begin{equation}
    \dot\xi_V\leq (F-Gc\transpose_V)\xi_V,
  \end{equation}
  {in which the last element correspond to the condition in \eqref{L2}. Subclaim~\ref{it:der}} can be proved similar to the claim 1 in Proposition \ref{prop:ECBF} where $q_i\in \real{+}$ for $i=1,\hdots,r$ is a pole location of $\mu_V = -c_V^T\xi_V$. Applying Gronwall's comparison lemma \cite{gronwall1919note}, we then conclude that $\lim_{t\to\infty} \xi_V=0$, which, in particular, implies subclaim~\ref{it:Vlimit}. Finally, subclaim~\ref{it:xlimit} can be shown using ~\ref{it:Vlimit} in combination with \eqref{L1} and standard arguments from Lyapunov theory \cite[Chapter 4] {khalil2002nonlinear}. 
\end{proof}

Note that this result can be applied to any point other than the origin with a simple change of coordinates.
\subsection{Polygonal environment decomposition}
We assume a compact polygonal environment $\cP\subset \real{n_x}$, not necessarily simply-connected, decomposed in a finite number of convex cells $\{\cX_{i,\textrm{pos}}\}$, such that $\bigcup_{i} \cX_{i,\textrm{pos}}=P$, and set $\cX_{i,\textrm{pos}}$ is a polytope defined by linear inequality constraints of the form $A_{\xpos,i}\transpose x\leq b_{\xpos,i}$.

Our goal is to design a different linear feedback controller $u$ for each cell $\cX_{i}$. The feedback signal used by the controller will be based on linear relative measurements with respect to a set of \emph{landmarks}.
\begin{definition}
  A landmark is defined as a point $\hat{y}\in \real{d}$ whose location is known and fixed in the environment.
\end{definition}
For each convex section $\cX_{i}$, we have a finite number of landmarks. In this paper, we choose the landmarks as the vertices of the convex section $\cX_{i}$, although this choice does not make any difference in terms of the actual method.

\subsection{High-level planning}\label{planning}
We consider two overall objectives for the controller design:
\begin{lenumerate}{O}
\item \label{it:point-stabilization} Point stabilization: given the stabilization point (where $\dot{x}=0$) in the environment and starting from any point, we aim to converge to the stabilization point (e.g. Fig.~\ref{1-1}).
\item\label{it:patrolling} Patrolling: starting from any point, we aim to patrol the environment by converging to a path, and then traversing the same path (e.g. Fig.~ \ref{osc1}).
\end{lenumerate}

To specify the convergence objective for each controller $u$, we first abstract the polygonal environment $P$
into a graph $\cG=(\cV,\cE)$, where each vertex $i \in \cV$ represents a cell $\cX_i$ in the partition of $P$, and an edge $(i,j)\in \cE$ if and only if cells corresponding to $i$ and $j$ have a face in common.

  In the case of the point stabilization objective \ref{it:point-stabilization}, the stabilization point is one of the vertices of the graph and if the stabilization point is in the middle of the cell, without loss of generality, we can decompose the cell into  new convex cells such that the stabilization point is one of the vertices of the new cells. Then, we add one vertex to the set $\cV$, which will be the stabilization point and also, we add edges between the new vertex and any cell that has a face in common with the cell includes the stabilization point to the set $\cE$.


For each cell, we then select one \emph{exit edge} (a pointer) such that, when considered together, all such edges provide a solution in the abstract graph $\cG$ to the high level objective. For instance, in the case of objective \ref{it:point-stabilization}, the exit edge of each cell will point in the direction of the shortest path toward the vertex of the stabilization point. In the case of objective \ref{it:patrolling}, following the exit edges will lead to a cyclic path in the graph.

To give an example, the polygonal environment in Fig.~\ref{g1} is converted to the connected graph in Fig.~\ref{g3} based on the cell decomposition of the environment in Fig.~\ref{g2}. Starting from the first node in Fig.~\ref{g3} which is shown by the green point, we find the path from the start node to the equilibrium node shown by the red point, through the path planning algorithms (e.g. using Dijkstra's algorithm). Regarding to that path, we define the \textit{exit face} as the face of the convex section the path moves through and based on that we design the controller.
\begin{figure}[t!]
  \centering
  \subfloat[]{\label{g1}{\includegraphics[width=2.5cm]{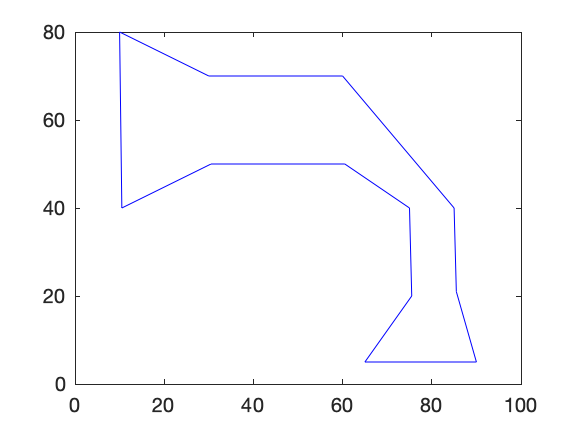} }}%
  \subfloat[]{{\label{g2}\includegraphics[width=2.5cm]{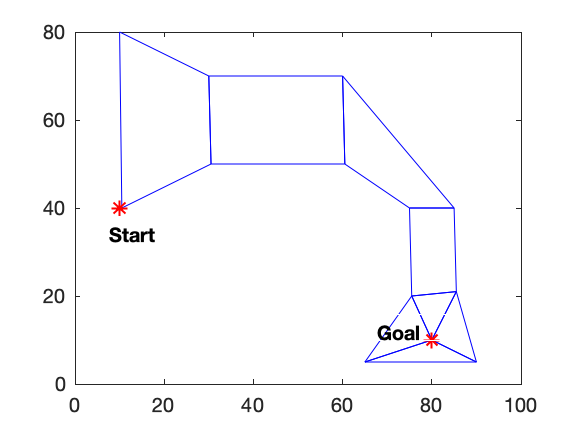} }}%
  \subfloat[]{{\label{g3}\includegraphics[width=2.6cm]{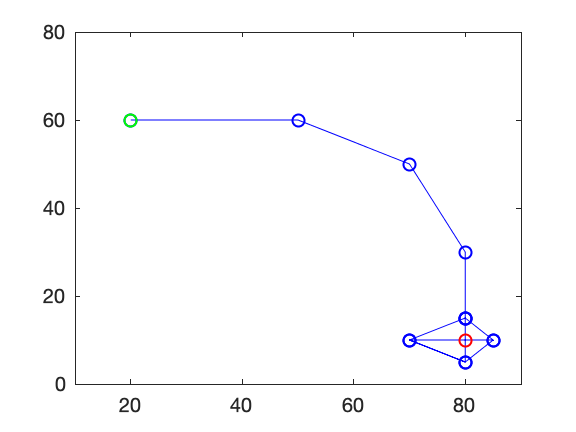} }}%
  \caption{The polygonal environment in Fig~.\ref{g1} is decomposed to 8 convex sections Fig~.\ref{g2} and the corresponding graph is shown in Fig~. \ref{g3}}
  \label{fig_env}
\end{figure}
\begin{definition}\label{exit_dir}
  For each cell $\cX_i$ in the decomposition of the environment, we define an \emph{exit face} $\cP_{exit}$ or, respectively, \emph{stabilization point} $\cP_{exit}=\{x_g\}$ to be the face or, respectively, vertex corresponding to the \emph{exit edge} in the abstract graph $\cG$. The \emph{exit direction} $z$ is an inward-facing normal or, respectively, direction of $\cP_{exit}$.
\end{definition}

In this work we desire to design a controller for each convex section of the environment that drives the system in the exit direction toward the exit face or the stabilization point, while avoiding the boundary of the environment.

Overall, thanks to the high level planning in the abstract graph $\cG$, and the controller design in each cell $\cX_i$ (explained in the sections below) the system will traverse a sequence of cells to reach a given equilibrium point, or achieve a periodic steady state behavior (examples in Section~\ref{numerical results}) according to the desired objective.


\section{PROBLEM SETUP}\label{problem setup}

The goal of this section is to synthesize a robust controller for a convex cell $\cX$ (with respect to previous sections, we dropped the subscript $i$ to simplify the notation) where  $\cX=\cX_{\textrm{dyn}}\cap \cX_{\textrm{pos}}$.
According to Definition \ref{xposdyn} we divide $x$ into two parts, $\xpos$ and $\xdyn$. We assume that the agent has direct access to $\xdyn$, but for $\xpos$ the agent can only measure the the relative displacements between the robot's position $\xpos$ and the landmarks in the environment, which corresponds to the output function
\begin{equation}
  y=(Y-\xpos\vct{1}\transpose)^\vee,
\end{equation}
where $Y\in\real{d\times n_l}$ is a matrix of landmark locations and  $A^\vee$ represents the vectorized version of a matrix $A$. Our goal is to find a feedback controller that, given $y$, provides an input $u$ that drives the system toward an exit face or vertex of $\cX$ while avoiding obstacles (non-exit faces of $\cX$). Note that the landmarks do not necessarily need to belong to $\cX$. We assume \eqref{sys1} is controllable and choose a controller of the form
\begin{equation}\label{u}
  u(K_1,K_2)=K_1y +K_2\xdyn,
\end{equation}
where $K_1 \in \real{n_u\times dn_l}$ and $K_2 \in \real{n_u\times(n_x-d)}$ are feedback gains that need to be designed. Note that $Y-\xpos\vct{1}\transpose$ gives a matrix where each column is the relative displacement between each landmark and the current position of the system; as such, we are looking for a controller that feeds back linear combinations of these displacements.
From the distributivity property of vectorization we can write $u$ as
\begin{equation}\label{UwithX}
\begin{aligned}
  u(K_1,K_2)&=K_1Y^\vee-K_1(\vct{1}_{nl} \otimes I_d)\xpos+K_2\xdyn\\&=K_1Y^\vee+K_xx,
  \end{aligned}
\end{equation}
where $K_x\in \real{n_u\times n_x}$ and $K_x=\begin{bmatrix}-K_1(\vct{1}_{nl} \otimes I_d)& K_2\end{bmatrix}$.
\begin{remark}
  In general, our framework can handle general linear output $y=Cx+D$, but we focus here on the path planning application.
\end{remark}
\subsection{Control Barrier Function}\label{sec_cbf}
Let $A_{h,i}\in\real{1\times n_x}$ belongs to the union of all rows of $A_{\xdyn}$ and $A_{\xpos}$ except the one row defining the exit face, and $b_{h,i}\in\real{}$ we define the following candidate ECBF:
\begin{equation}\label{cbf_f}
  h_i(x)={A_{h,i}}x+b_{h,i}
\end{equation}
where $i=\{1,\hdots,s_p+s_d\}$ such that $s_p$ denotes the number of faces of $\cX$ except the one associated to an exit face (or all of them in the case of a stabilization point) and $s_d$ donates the number of number of boundaries to limit $\xdyn$. 
We define $s_x=s_p+s_d$.
\subsection{Control Lyapunov Function} \label{sec_clf}
To stabilize the system, we define the Lyapunov function $V(x)$ for cell $\cX$ as,
\begin{equation}\label{clf_f}
  V(x)=z^T x+b_V, \;\;
\end{equation}
where $\Ppos z\in \real d$ is the exit direction for the cell $\cX$ (see Definition~\ref{exit_dir}), and $\Pdyn z=0$, and $b_V \in \real{}$ is chosen such that the function reaches its minimum $V(x)=0$ when $x$ is in the exit face ($V(x)<0$ correspond to points outside the cell). Note that this Lyapunov function represents, up to a constant, the distance $d(\xpos,\cP_{exit})$ between the current system position and the exit face. When the exit face reduces to an exit point $\cP_{exit}=x_{exit}$, the Lyapunov function states the distance $d(\xpos,x_{exit})$ between the current position and exit point, up to a constant, and the minimum of $V(x)=0$ reaches when $x$ is identical to the exit point $x_{exit}$. 

\begin{remark}
  The function $V(x)$ can be defined as a function of the vertices of the exit face instead of its normal.
 For instance, in $\real{2}$, we have
   \begin{equation}
    V(x)=\det(\bmat{v_1-v_0 & \xpos})
  \end{equation}
  where $v_0,v_1$ are two distinct points (e.g., vertices) in the exit face (with their order determining the correct sign in $V(x)$). Based on the same idea, in $\real{3}$, $  V(x)=\det(\bmat{v_1-v_0 &v_2-v_0& \xpos})$
  where $v_0,v_1,v_2$ are three distinct points in the exit face (e.g., three vertices of the exit plane) and respectively . This concept can be generalized to any dimension.
\end{remark}
\subsection{Finding the Controller by Robust Optimization}
Our goal is to find controllers $u$ (more precisely, control gains $K$) that maximize the motion of the robot toward the exit face, while avoiding the boundary of the environment.
Using the CLF-CBF constraints reviewed in Section~\ref{notation}, we encode our goal in the following feasibility problem:

\begin{equation}\label{opt-feasibility}
  \begin{aligned}
    \textrm{find} \;\;{K}\\
    \textrm{s.t.}:&-( \cL_{Ax}^{r}h_i(x)+\cL_B\cL_{Ax}^{r-1}h_i(x)u+c\transpose_b\xi_{bi}(x))\leq 0,\\
    &\cL_{Ax}^{r}V(x)+\cL_B\cL_{Ax}^{r-1}V(x)u+c\transpose_V\xi_V(x)\leq 0,\\
    &u\in\cU,\\
    & \forall x\in \cX,\;\;i=\{1,\hdots,s_h\}.
  \end{aligned}
\end{equation}

In practice, we aim to find a controller that satisfies the constraints in \eqref{opt-feasibility} with some margin, hence we focus on the following robust optimization problem:

\begin{equation}\label{opt_robust}
  \begin{aligned}
    \min_{K,S_l,S_b}\;&w_b\transpose S_b+w_lS_l\\
    \textrm{s.t.}: &-(\cL_{Ax}^{r}h_i(x)+\cL_B\cL_{Ax}^{r-1}h_i(x)u+c\transpose_b\xi_{bi}(x))\leq S_{bi},\\
    & \;\;\;\;\cL_{Ax}^{r}V(x)+\cL_B\cL_{Ax}^{r-1}V(x)u+c\transpose_V\xi_V(x)\leq S_l,\\
    &  \;\;\;\;S_l,S_b \leq 0, u\in\cU,\\
    &  \;\;\;\;\forall x\in\cX,\;\;i=\{1,\hdots,s_h\}.
  \end{aligned}
\end{equation}
Note that the constraints in \eqref{opt_robust} need to be satisfied for all $x$ in the cell $\cX$, i.e., the same control gains should satisfy the CLF-CBF constraints at every point in the cell. We handle this type of constraint by rewriting \eqref{opt_robust} using a min-max formulation, where \eqref{opt_robust} is equivalent to,
\begin{equation}\label{opt_org}
  \begin{aligned}
    &\min_{K,S_l,S_b}\;w_b \transpose S_b+w_l S_l\\
    &\textrm{s.t.}:\\
    &\begin{bmatrix} \max_x -(\cL_{Ax}^{r}h_i(x)+\cL_B\cL_{Ax}^{r-1}h_i(x)u+c\transpose_bh_i(x))\\
      s.t \;\;\; x\in \cX,\;u \in \cU
    \end{bmatrix}
    \leq S_{bi},\\
    &\begin{bmatrix} \max_{x} \cL_{Ax}^{r}V(x)+\cL_B\cL_{Ax}^{r-1}V(x)u+c\transpose_VV(x)\\
      s.t \;\;\; x\in \cX,\;u \in \cU
    \end{bmatrix}
    \leq S_l,\\
    &S_l,S_b \leq 0, \;\;i=\{1,\hdots,s_h\}.
  \end{aligned}
\end{equation}

the weights $w_b \in \real{s_h}$ and $w_l\in \real{}$ are user-defined constants defining the trade-off between  the barrier functions and Lyapunov function constraints: larger $w_b$ implies solutions moving away from the walls, while larger $w_l$ implies solutions moving faster toward the exit face. 
Now we compute the time $r$-th order derivative of $h_i(x)$ such that
\begin{equation}\label{Lieofh}
    h_i^{r}(x)=\cL^{r}_{Ax}h_i(x)+\cL_B\cL_{Ax}^{r-1}h_i(x)u
\end{equation}
Combining $h(x)$ from \eqref{cbf_f} and \eqref{Lieofh} implies
\begin{equation}\label{hr}
    h_i^{r}(x)
    =A_{h,i} A^{r-1}(Ax+Bu)=A_{h,i} A^{r-1} \dot{x},
\end{equation}
where $A^r$ represents the $r$-th power of $A$. Similar to $h_i^{r}(x)$, the time derivative of $V(x)$ is defined as
\begin{equation}\label{zr}
  V^{r}(x) =z A^{r-1}(Ax+Bu)=z A^{r-1}\dot{x}.
\end{equation}

Substituting the Lie derivatives \eqref{hr} and \eqref{zr} in \eqref{opt_org} results

\begin{equation}\label{opt_min_max}
  \renewcommand{\arraystretch}{1.5}
  \begin{aligned}
    &\min_{K,S_l,S_b}\; w_b\transpose S_b+w_l S_l\\
    &\textrm{s.t.}:\\
    &\begin{bmatrix}
      \underset{x}{\max}-(A_{h,i}A^r+A_{h,i}A^{r-1}B K_x+A_{h,i}{c_b})x\\
      \textrm{s.t.}:\;\;A_x x \leq b_x\\
    \end{bmatrix}
    \leq \\&
    \quad \quad  \quad \quad \quad \quad \quad  \quad \quad  \quad \quad
    S_{bi}+c_bb_{h,i}+{A_{h,i}}A^{r-1}B K_1Y^\vee\\
    &\begin{bmatrix}\underset{x}{\max}(z\transpose A^{r}+z\transpose A^{r-1}B K_x+z\transpose {c_l})x \\
      \textrm{s.t.}:\;\;A_x x \leq b_x\\
    \end{bmatrix}
    \leq \\&
    \quad \quad  \quad \quad \quad \quad \quad  \quad \quad  \quad \quad S_l-c_lb_V-z^T A^{r-1}B K_1Y^\vee,\\
    &\begin{bmatrix}
      \underset{x}{\max}(A_{uj}K_x)x\\
      \textrm{s.t.}:\;\;A_x x \leq b_x\\
    \end{bmatrix}\leq b_{uj}-A_{uj}K_1Y^\vee\\
    & S_b,\;S_l \leq 0, \;i=\{1,\hdots,s_h\},j=\{1,\hdots,n_u\}
  \end{aligned}
\end{equation}
In \eqref{opt_min_max}, we have a bi-level optimization problem with constraints that are given themselves by other optimization problems. As all constraints and objective function are linear and $\cX$ is a convex set, \eqref{opt_min_max} and inner maximization problems are linear programming problem so we can change the min-max problem \eqref{opt_min_max} to min-min problem by replacing the inner maximization problems with their dual forms,
\begin{equation}\label{opt_dual}
  \begin{aligned}
    &\min_{K,S_l,S_b}\; w_b\transpose S_b+w_lS_l\\
    & s.t.:\\
    &\begin{bmatrix}
      \min_{\lambda_{b}} \lambda_{b}\transpose b_x \\
      \textrm{s.t.}:\\
      A_x\transpose\lambda_{b}=-(A_{h,i}A^{r}+A_{h,i}A^{r-1}B K_x+A_{h,i}{c_b})\transpose\\
      \lambda_{b}\geq 0,
    \end{bmatrix} \leq\\& \quad  \quad \quad  \quad \quad  \quad \quad  \quad \quad \quad
    S_{bi}+c_bb_{h,i}+{A_{h,i}}A^{r-1}B K_1Y^\vee\\
    &\begin{bmatrix}
      \min_{\lambda_l}\lambda_l  \transpose b_x\\\textrm{s.t.}:\\
      A_x\transpose \lambda_l
      =(z\transpose A^{r}+z\transpose A^{r-1}B K_x+z\transpose {c_l})\transpose\\
      \lambda_l \geq 0
    \end{bmatrix}\leq \\& \quad \quad  \quad \quad  \quad \quad  \quad \quad  \quad \quad \quad S_l-c_lb_V-z^TA^{r-1}B K_1Y^\vee\\
    &\begin{bmatrix}
      \min_{\lambda_l}\lambda_u  \transpose b_x\\\textrm{s.t.}:\\
      A_x\transpose \lambda_u
      =(A_{uj}K_x)\transpose\\
      \lambda_u \geq 0
    \end{bmatrix}\leq b_{uj}-A_{uj}K_1Y^\vee\\
    &S_b,\;S_l \leq 0, \;i=\{1,\hdots,s_h\},j=\{1,\hdots,n_u\}
  \end{aligned}
\end{equation}
where min-min problem \eqref{opt_dual} is equivalent to the minimization problem,
\begin{equation}\label{opt_min}
  \begin{aligned}
    \min_{K,S_l,S_b,\lambda_l,\lambda_b}\; &w_b\transpose S_b+w_lS_l\\
    s.t.:\;
    &\lambda_b \transpose b_x \leq  S_{bi}+c_bb_{h,i}+{A_{h,i}}A^{r-1}B K_1Y^\vee\\
    &\lambda_l\transpose  b_x\leq S_l-c_lb_V-z\transpose A^{r-1}B K_1 Y^\vee\\
    &\lambda_u\transpose  b_x\leq b_{uj}-A_{uj}  K_1 Y^\vee\\
    &A_x\transpose\lambda_b=-(A_{h,i}A^{r}+A_{h,i}A^{r-1}B K_x+A_{h,i}c_b)\\
    &A_x\transpose\lambda_l=(z\transpose A^{r}+z\transpose A^{r-1}B K_x+z\transpose {c_l})\transpose\\
    &A_x\transpose\lambda_u=(A_{uj}K_x)\transpose\\
    &\lambda_l,\lambda_b,\lambda_u\geq 0 \;\;S_b,S_l \leq 0\\& i=\{1,\hdots,s_h\}, j=\{1,\hdots,n_u\}
     \end{aligned}
\end{equation}
In the following we prove that the feasible optimal solution for \eqref{opt_min_max} is also the feasible optimal solution for \eqref{opt_min}.
\begin{remark}\label{strong_duality}
  By strong duality \cite[Theorem 4.4]{LP}, if a linear programming problem has an optimal solution, so does its dual, and the respective optimal costs are equal.
\end{remark}
This remark allows us to prove the following.
\begin{lemma}\label{lem1}
  Optimization problems \eqref{opt_min_max} and \eqref{opt_dual} have the same feasible optimal solution.
\end{lemma}
\begin{proof}
  The optimization problems in \eqref{opt_min_max} and \eqref{opt_dual} have the same objective functions. Constraints in \eqref{opt_min_max} are in the form of LP optimization problem and the constraints in \eqref{opt_dual} are the duals. According to the Remark \ref{strong_duality}, the optimal cost of constraints in \eqref{opt_min_max} and \eqref{opt_dual} are equal and result the same constraints with the same objective functions which implies \eqref{opt_min_max} and \eqref{opt_dual} have the same optimal solution.
\end{proof}
\begin{lemma}\label{lem2}
  Optimization problems \eqref{opt_dual} and \eqref{opt_min} have the same feasible optimal solution.
\end{lemma}
\begin{proof}
  Assume we have an optimal solution for \eqref{opt_min}, then the solution is also feasible for \eqref{opt_dual} and the objective costs are the same. In the same way, if we have an optimal solution for \eqref{opt_dual}, so there must exist dual variables for inner optimization problem in \eqref{opt_dual} which are also feasible for \eqref{opt_min} and result in the same objective cost\cite{robust_opt}.
\end{proof}

\begin{theorem}
From Lemma \ref{lem1} and Lemma \ref{lem2}, the optimization problems \eqref{opt_min_max}, \eqref{opt_dual} and \eqref{opt_min} are equivalent.
\end{theorem}
\begin{proposition}
Solving \eqref{opt_min} and assuming optimal $S_l$ is strictly less than zero, then the trajectory exit from the cell in finite time. 

 \begin{proof}
    For the first order system $c_V>0$ is scalar and $V(x)$ is positive definite, define maximum distance from the exit face as 
    \begin{equation}
        d_{max}=max\{V(x)|A_xx\leq b_x\}.
    \end{equation}
    For the first order system, the CLF constraints in \eqref{opt_dual} implies
    \begin{equation}
        \dot{V(x(t))}+c_VV(x(t))\leq S_l
    \end{equation}
    where $Sl< 0$ and $g_1=c_VV(x(t))\geq0$ and results in
    \begin{equation}
        \dot{V(x(t))}\leq S_l-g_1=g_2, g_2\leq0
    \end{equation}
Solving the above differential equation shows $V(x(t))\leq V(x_0)+g_2t $. To pass the exit face we need to have $V(x(t))=0$ as the $V(x(t))$ shows the distance from the exit face, so $t_{exit}\leq -\frac{d_{max}}{g_2}$ and the controller reaches the exit face in finite time as $-\frac{d_{max}}{g_2}$ has a finite value. 
 \end{proof}
\end{proposition}
\section{STATIONARY POINT}\label{stationary}
Consider the stabilization objective \ref{it:point-stabilization} defined in Section \ref{planning}, and let $x_g$ be the stabilization point in $\cX$, i.e., the exit vertex in $\cP_{exit}$ (see Definition~\ref{exit_dir}). In this section we provide sufficient conditions that show the controllers synthesized with our proposed method indeed introduce an equilibrium point at $x_g$.
Before proceeding, we need the following. We use $\stack()$ to denote the operator that stacks vertically all its matrix arguments.
\begin{fact}\label{fact:h-cone}
  Let $A_{h,exit}$ be the matrix whose rows are the row vectors in the set $\{A_{h,i}: h_i(x_g)=0\}$. Then $z$ belongs to the proper cone $\{v: A_{h,exit}v\geq 0\}$. 
\end{fact}
This fact is intuitive given Definition~\ref{exit_dir}: $A_{h,\textrm{exit}}$ represents the normals of the active constraints at the stabilization point, and $z$ needs to be inward-pointing. Note that the rows or $A_{h,\textrm{exit}}$ are a subset of the rows of $A_{\xpos}$. We can now state the main result of this section.
\begin{proposition}\label{goalPoint}
  Assume the pair $(A,\stack(A_{h,\textrm{exit}},z\transpose))$ is observable and that all $h_i$ and $V$ have the same relative degree $r$. Then, any solution to the min-max problem \eqref{opt_robust} (or, equivalently, the linear program \eqref{opt_min}) guarantees that $\dot{x}=0$ when $\xpos=x_g$.
\end{proposition}
Note that the assumption about having a homogeneous relative degree is reasonable, since $z\transpose$ and $A_{h,\textrm{exit}}$ all essentially represent generic planes in the environment.

\begin{proof}
  Any feasible controller must satisfy the CBF and CLF constraints in~\eqref{opt_robust}. As discussed above, we have $V(x_g)=0$ for the Lyapunov function, and $h_i(x_g)=0$ for the constraints corresponding to $A_{h,\textrm{exit}}$. Recalling that $\cL_A^0V=V$, and using the fact the CLF constraint in \eqref{opt_robust} implies Proposition~\ref{prop:ECLF}, claim~\ref{it:der}, we have that $\cL_A^j V=z\transpose A^j\dot{x}\leq -c_v \cL_A^{j-1} V= z\transpose A^{j-1}\dot{x}$ for all $0\leq j \leq r$. A similar argument with the CBF constraint in \eqref{opt_robust} and Proposition~\ref{prop:ECBF} implies that $A_{h,\textrm{exit}}A^j\dot{x}\geq -c_hA_{h,\textrm{exit}}A^{j-1}\dot{x}$.
  From Fact~\ref{fact:h-cone}, we have that the sets described by $A_{h,\textrm{exit}}v\geq 0$ and $z\transpose v\leq 0$ intersect only at the point $v=0$; hence, $\stack(A_{h,\textrm{exit}},z\transpose) A^j\dot{x}=0$ for all $0\leq j \leq r-1$, which can be compactly described as $\cO_{A} \dot{x}=0$, where $\cO_{A}$ is the observability matrix from the pair $(A,\stack(A_{h,\textrm{exit}},z\transpose))$. Since the latter is observable, $\cO_{A}$ is full rank, and hence $\dot{x}=0$ as claimed. 
\end{proof}
Intuitively, the proof shows that the CLF and CBF constraints fix $\xpos$ to $x_g$, which together with the observability implies that also $\xdyn=0$.

\section{NUMERICAL EXAMPLES}\label{numerical results}
In this section, we apply our proposed method to two non-simply connected environments to find a output-feedback controller, then we deform the environment and implement the same controller to represent the robustness of the controller. We apply our method to the fist and second integrator systems to achieve two planning objectives of our algorithm. The first example considers the point stabilization \ref{it:point-stabilization} objective and the second example shows the patrolling \ref{it:patrolling} objective (The choice of the fist and second integrator systems is independent of the planning objective and complexity of the environment).
\begin{figure}[t]
  \centering
  \subfloat[Polygonal environment]{\label{1-1}{\includegraphics[width=4cm]{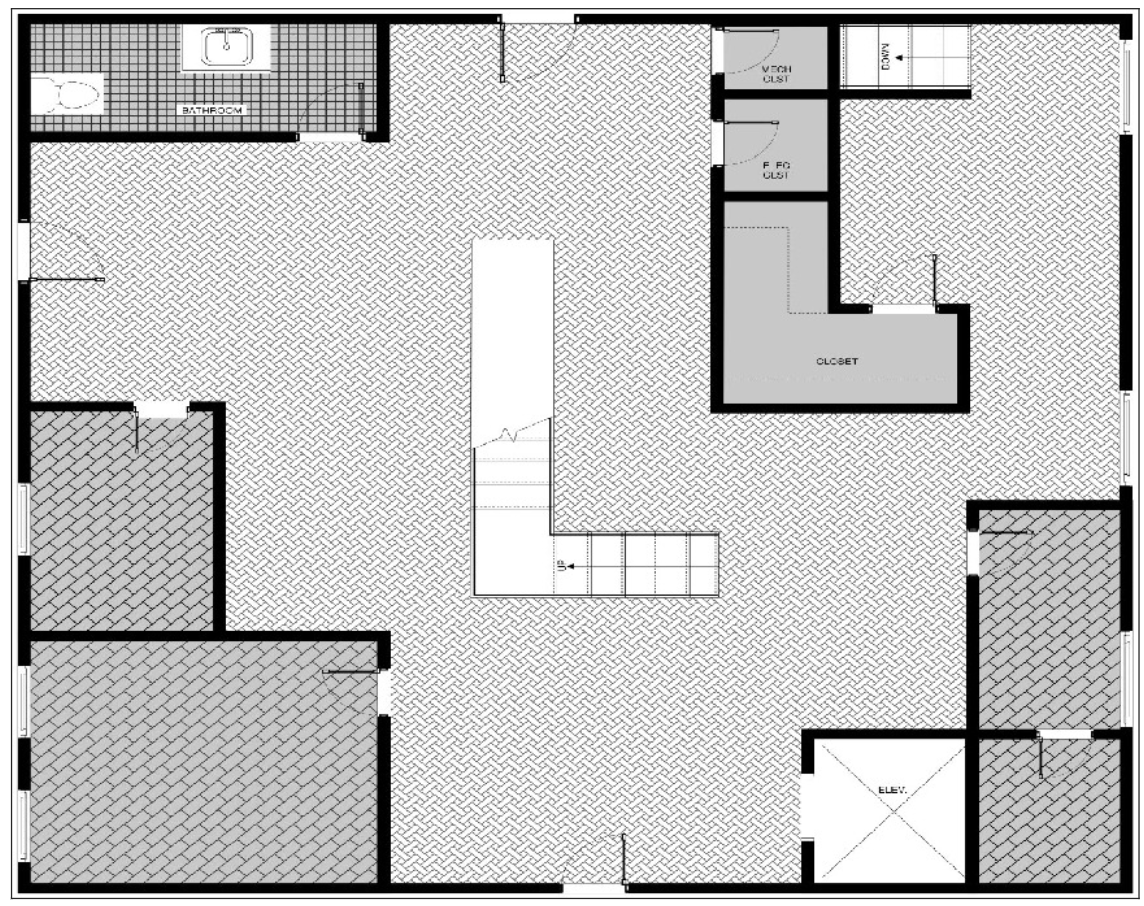} }}
  \subfloat[Decomposed environment]{{\label{1-2}\includegraphics[width=4cm]{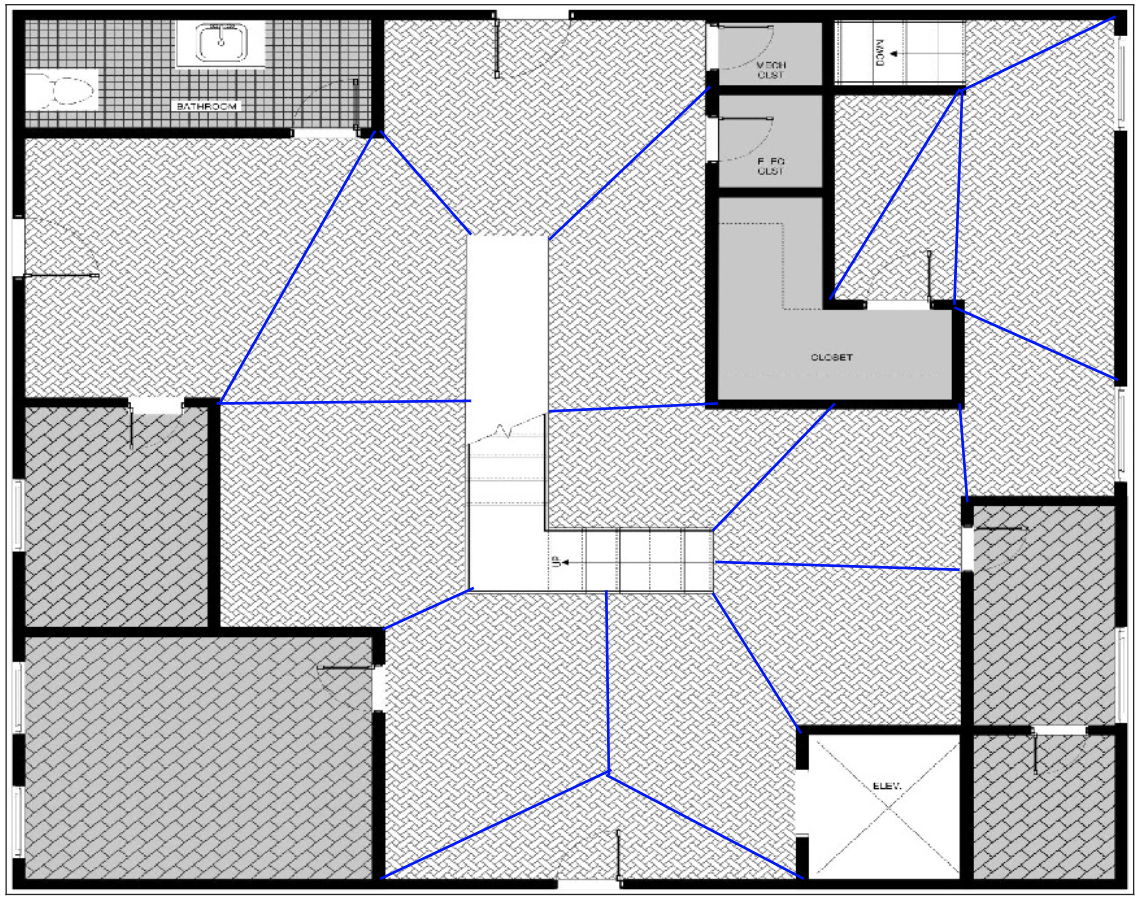} }}%
  \caption{Non-simply connected environment is decomposed to a set of convex cells. }
  \label{fig_env2}
\end{figure}
\subsection{First Order Controller}
Given the environment in Fig \ref{fig_env2}, we want to find a set of controllers to move the agent from three different entrances (start points) of a building to the exit door (goal point). To achieve this goal, we decompose the layout of the floor into 16 convex cells. We choose $c_b=c_v=0.5$.
Given the decomposed environment, we find a controller for each cell individually and move the agent from the start points. In this example, we assume the landmarks for each cell are same as the vertices of the cell and we define as $z$ the exit direction of the cell.  Solving the optimization problem \eqref{opt_min} finds the optimal $K_1$ for each cell which implies the optimal controller. Entering the building from different doors in Fig.~\ref{org}, the first order controller moves the agent from to the exit door without violating any constraints. Assume the layout of the building is changed due to the construction purposes in Fig.~\ref{def} and the agent enters the building aims to proceed to the exit door. When the layout deformed, the convex cells change as the position of landmarks and given the old optimal $K_1$ from the original layout and new position of landmarks, the agent is able to proceed to the exit door starting from different entrances and meeting all the safety and stability constrains.  
\begin{figure}%
  \centering
  \subfloat[
  Original environment]{\label{org}{\includegraphics[width=4cm]{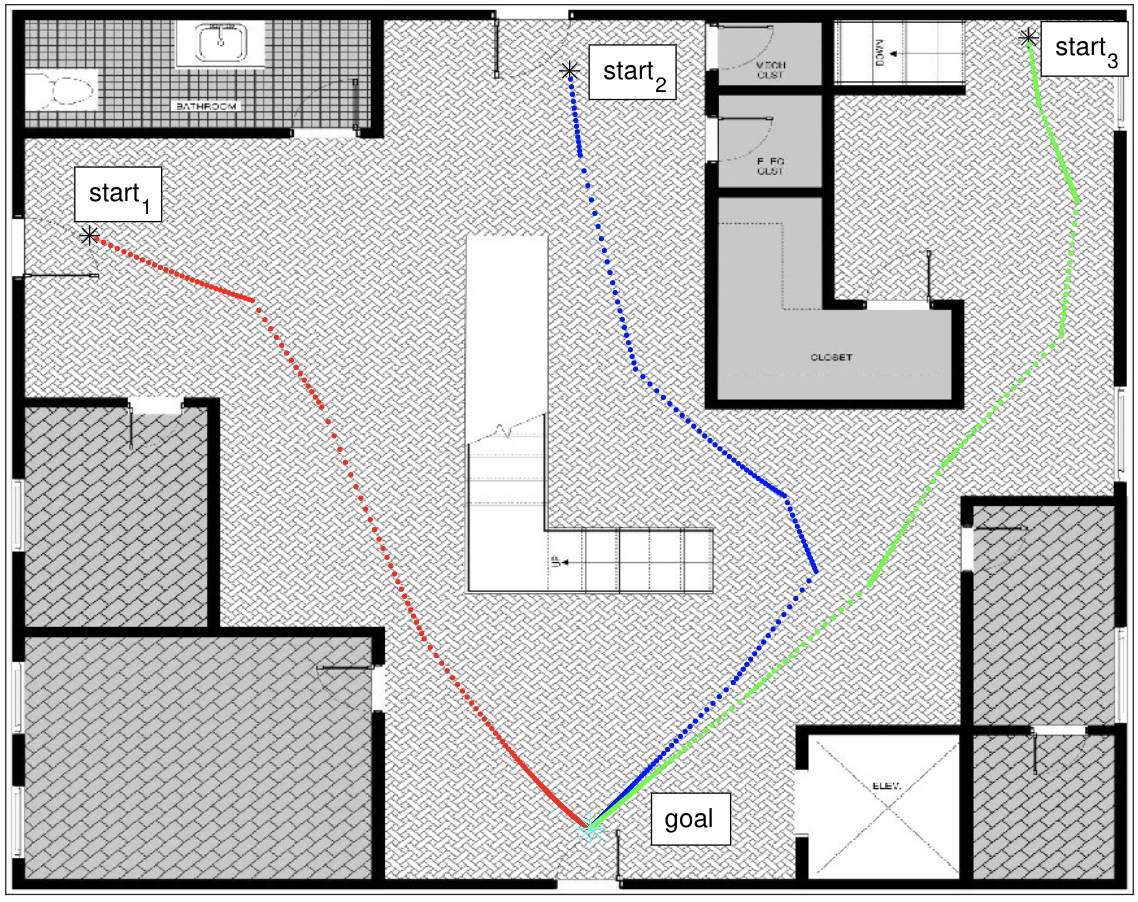} }}
  \subfloat[Deformed environment]{{\label{def}\includegraphics[width=4cm]{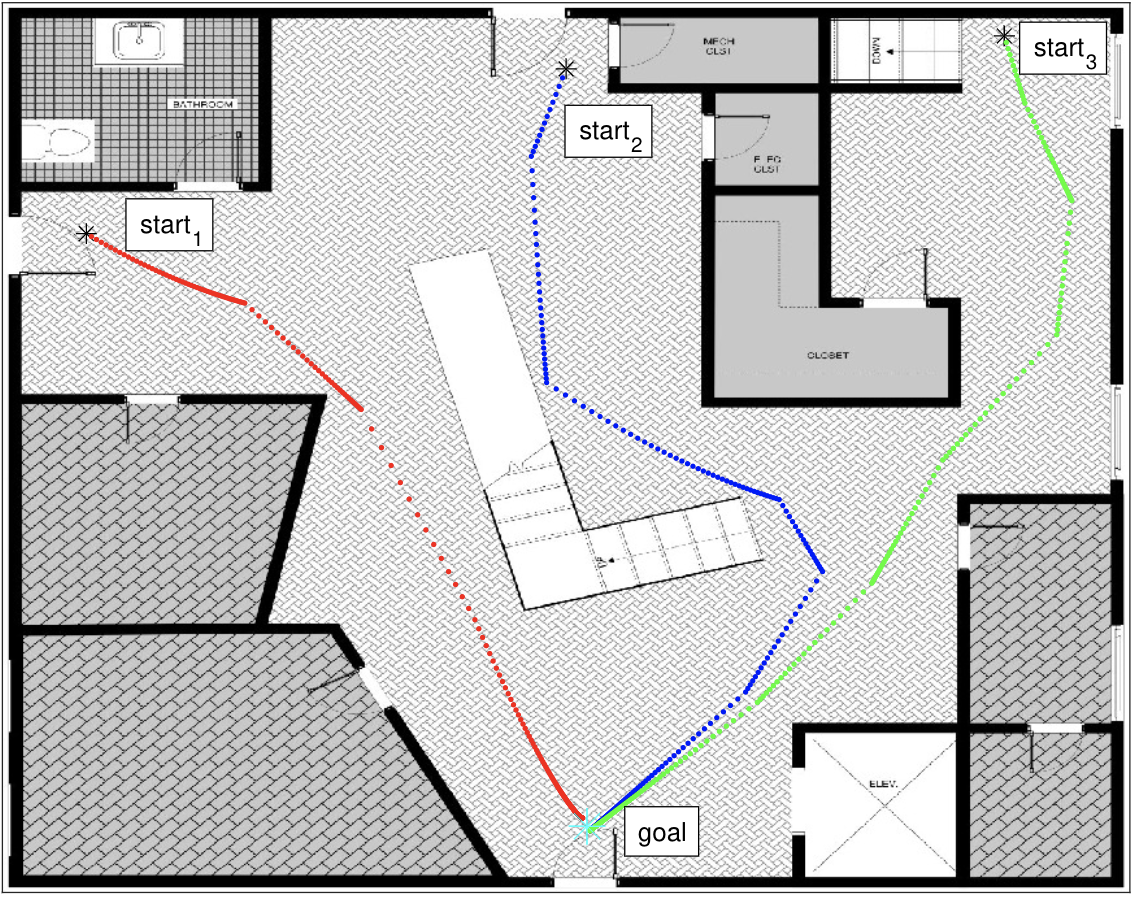} }}%
  \caption{In this example, we deform the layout and apply the first order controller of the original environment to the deformed layout. Despite the fact that we deform the layout significantly, the exact same original controllers generate successful trajectories.}%
  \label{sol1_reshape}%
\end{figure}
\subsection{Second Order Controller}
In this section, we find a controller for a second order system and similar to the first order system, we assume the landmarks are equivalent to the vertices of each cell and we donate $z$ the exit direction of the cell.
we apply our method to a non-simply connected environment.
The controller moves the agent from the start point at $[10,40]$ in Fig.~\ref{secObsOrg}, we choose $c_b=[1,1]\transpose$ and $c_V=[1,1]\transpose$ for all $i=\{1,\hdots,s_h\}$. Then, we enlarge the obstacle and apply the same controller to the agent in Fig.~\ref{secOdsEng}. Our method guarantees that the agent moves through the environment completely without violating safety and stability constraints and when obstacle rotates $\pi/4$ counterclockwise in Fig.~\ref{secObsRot}, the agent moves through the feasible path to cover all the environment with similar control gain $K_1$ and $K_2$.
\begin{figure}[t!]
  \centering
  \subfloat[Non-simply connected environment]{\label{osc1}{\includegraphics[width=4.2cm]{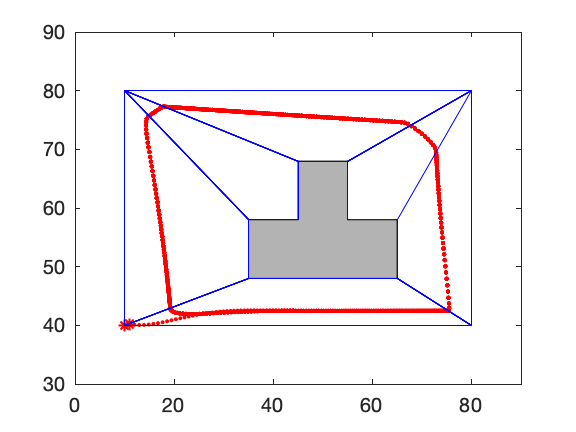} }}%
  \subfloat[Changes of the states versus time]{{\label{osc2}\includegraphics[width=4.2cm]{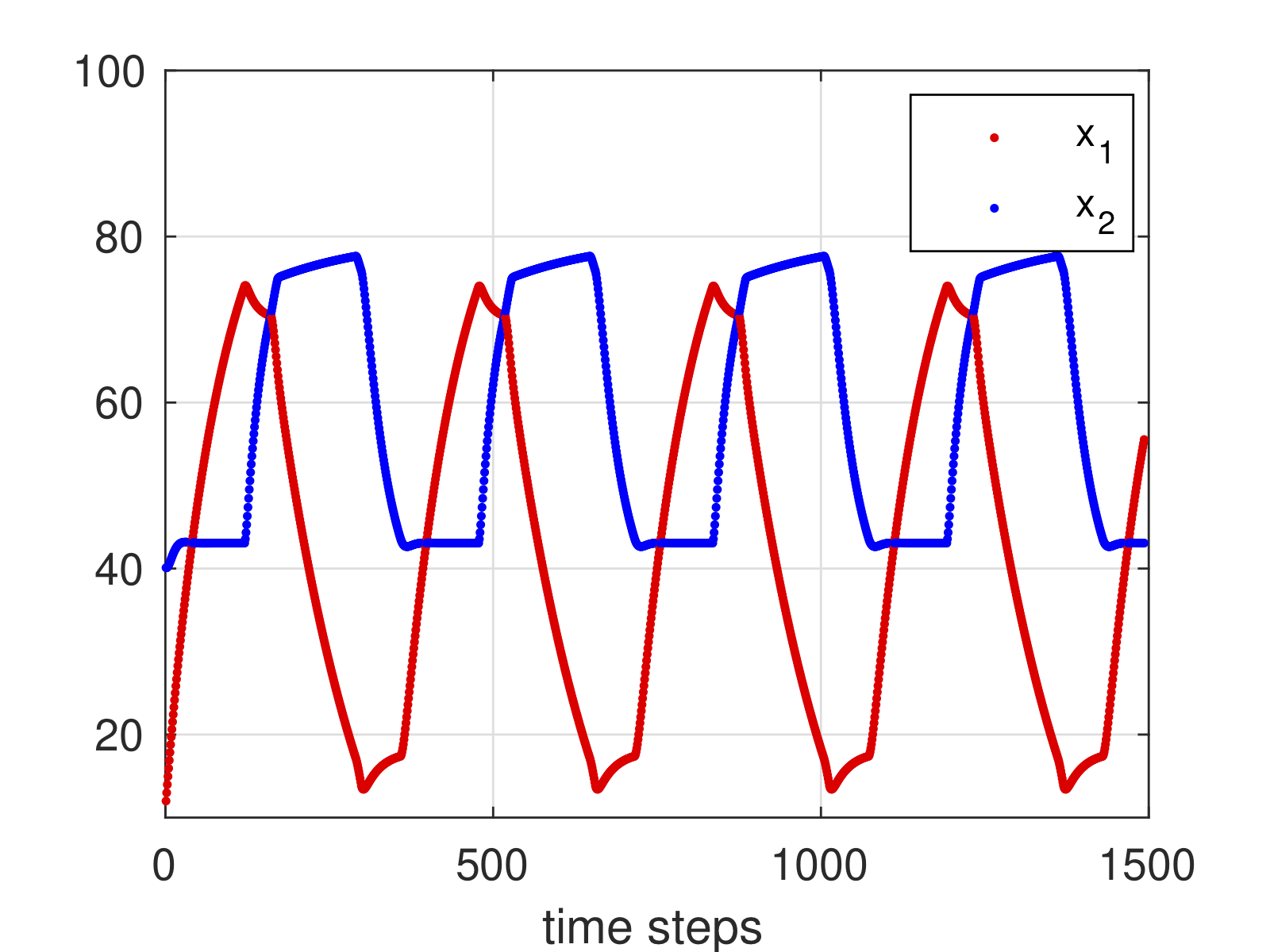} }}%
  \caption{ Fig \ref{osc1} is a non-simply connected environment and the gray polygon is an obstacle. In  Fig.~\ref{osc1} an agent starts from position (10,40) and continuously moves through all sections. In Fig.~\ref{osc2} $x_1$ and $x_2$ variation versus time is shown.}
  \label{secObsOrg}
\end{figure}
\begin{figure}[t]
\centering
  \subfloat[Enlarging the obstacle]{\label{secOdsEng}\includegraphics[width=4.5cm]{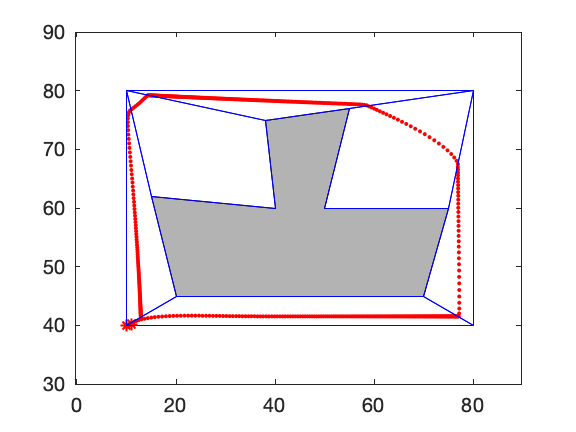} }
  \subfloat[Rotating the obstacle]{{\label{secObsRot}\includegraphics[width=4.5cm]{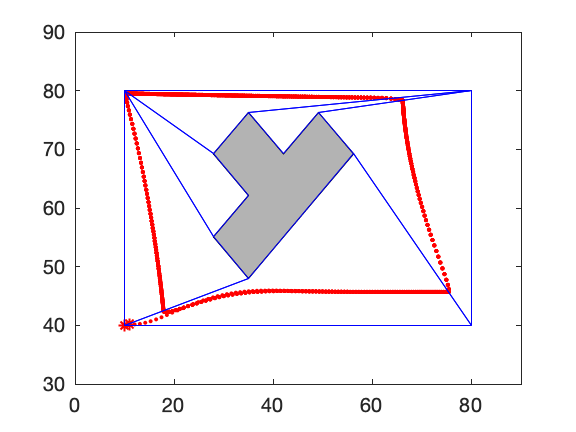} }}%
   \caption{In this two examples we deform the obstacle and apply the second order controller of the original environment to the new environment.}
   
\end{figure}
\section{CONCLUSIONS}\label{conclusion}
In this paper we proposed a novel approach to design a output-feedback controller with cell decomposition, through Linear Programming. We defined a controller such that it depends on the relative displacement measurements with respect to the landmarks of the convex cells and formed the min-max convex problem. Then we changed the min-max optimization problem to min-min optimization problem by forming the dual of the inner maximization problems and we found the controller which is robust to the significant changes of the environment. We validate our approach on different examples for the first and second order dynamic control systems. As presented, our current cell-focused approach has two limitations: First the controllers are discontinuous at the boundaries of the cells; this can be addressed by adding smoothness constraints between cells (at the price of solving a single large linear program instead of multiple ones). Second, we assume the environment is already discretized in convex cells; it is possible to relax this assumption by using sampling and Voronoi partitions. We plan to study these extensions in our future work. In addition, we aim to implement our method to constrained nonlinear systems based on \cite{girard2008motion}. Moreover, we also plan to formally investigate theoretical guarantees for the robustness of the synthesized controllers that we have empirically demonstrated in this paper. 
\bibliographystyle{ieee}
\bibliography{references}

\end{document}

%% file: preamble/common.tex
\input{preamble/commonNoTikz}
\input{preamble/graphicsTikz}
\input{preamble/robotics}


%% file: preamble/commonNoTikz.tex
\input{preamble/fixes}
\input{preamble/graphics}
\input{preamble/pagination}
\input{preamble/utilities}
\input{preamble/math}
\input{preamble/operators}

\input{preamble/notation}

\input{preamble/units}

\input{preamble/markupAndCommenting}


%% file: preamble/fixes.tex
\usepackage{fix-cm}
\usepackage{etex}

\usepackage{dblfloatfix}

\usepackage{nag}


%% file: preamble/graphics.tex
\makeatletter
\@ifpackageloaded{xcolor}{}{%
\usepackage[table,x11names,dvipsnames,svgnames]{xcolor}%
}
\makeatother

\usepackage{colortbl}

\usepackage{graphicx}
\usepackage{wrapfig}

\input{preamble/graphicsColors}


%% file: preamble/graphicsColors.tex
\definecolorset{RGB}{lyft}{}{Red,194,39,36;Sunset,202,53,33;Orange,205,68,20;Amber,200,117,42;Yellow,242,169,52;Citron,186,188,44;Lime,112,159,33;Green,56,139,31;Mint,45,118,56;Teal,52,133,135;Cyan,60,132,202;Blue,55,94,248;Indigo,64,13,247;Purple,115,42,248;Pink,176,25,145;Rose,176,32,75}

%% file: preamble/pagination.tex
\usepackage{cite}

\usepackage{microtype}

\usepackage{array}
\usepackage{multirow}
\usepackage{booktabs}
\usepackage{makecell} 

\newcommand{\myparagraph}[1]{\textbf{\emph{#1}}.}

\ifcsname labelindent\endcsname
\let\labelindent\relax
\fi
\usepackage[inline]{enumitem}

\usepackage{subfig}

\newenvironment{lenumerate}[2][]
{\begin{enumerate}[label=(#2\arabic*),leftmargin=0.2in,itemindent=0.15in,#1]}
{\end{enumerate}}

\setlist*[enumerate,1]{label={\itshape\arabic*)}}

\makeatletter
\newcommand{\paragraphswithstop}{%
\let\copyparagraph\paragraph%
\renewcommand\paragraph[1]{\copyparagraph{##1.}}%
}
\makeatother

\usepackage[framemethod=tikz]{mdframed}


%% file: preamble/utilities.tex
\usepackage{suffix}

\usepackage{environ}

\makeatletter
\newsavebox{\boxifnotempty}
\newcommand{\displayifnotempty}[3]{\sbox\boxifnotempty{#2}\setbox0=\hbox{\usebox{\boxifnotempty}\unskip}%
\ifdim\wd0=0pt
\else
 #1\usebox{\boxifnotempty}#3%
\fi%
}
\newcommand{\ifempty}[2]{\setbox0=\hbox{#1\unskip}%
\ifdim\wd0=0pt%
 #2%
\fi%
}
\newcommand{\ifnotempty}[2]{\setbox0=\hbox{#1\unskip}%
\ifdim\wd0>0pt%
 #2%
\fi%
}
\makeatother

\usepackage{algpseudocode}
\usepackage{algorithm}

\usepackage{scrlfile}

\makeatletter
\newcommand*\newstoreddef[1]{
  \BeforeClosingMainAux{%
    \immediate\write\@auxout{%
      \string\restoredef{#1}{\csname #1\endcsname}%
    }%
  }%
}
\newcommand*{\restoredef}[2]{
  \expandafter\gdef\csname stored@#1\endcsname{#2}%
}
\newcommand*{\storeddef}[1]{
  \@ifundefined{stored@#1}{0}{\csname stored@#1\endcsname}%
}
\makeatother

\usepackage{pageslts}
\NewEnviron{tee}{\BODY\typeout{Marker Tee [start] ^^J \BODY ^^JMaker Tee [end]}}


%% file: preamble/operators.tex
\newcommand{\real}[1]{\mathbb{R}^{#1}{}}

\newcommand{\bmat}[1]{\begin{bmatrix}#1\end{bmatrix}}

\newcommand{\transpose}{^\mathrm{T}}

\DeclarePairedDelimiter{\norm}{\lVert}{\rVert}

\newcommand{\vct}[1]{\mathbf{#1}}

\DeclareMathOperator{\stack}{stack}



%% file: preamble/notation.tex
\providecommand{\cC}{\mathcal{C}}

\providecommand{\cE}{\mathcal{E}}

\providecommand{\cG}{\mathcal{G}}

\providecommand{\cK}{\mathcal{K}}
\providecommand{\cL}{\mathcal{L}}

\providecommand{\cO}{\mathcal{O}}
\providecommand{\cP}{\mathcal{P}}

\providecommand{\cU}{\mathcal{U}}
\providecommand{\cV}{\mathcal{V}}

\providecommand{\cX}{\mathcal{X}}


%% file: preamble/units.tex
\usepackage{units}


%% file: preamble/markupAndCommenting.tex
\newcommand{\newcolorlabel}[2]{%
  \expandafter\newcommand\csname #1\endcsname[1]{%
    \colorbox{#2}{\color{white}\textsf{\textbf{##1}}}}%
}

\newcommand{\newcommenter}[2]{%
  \expandafter\newcommand\csname #1\endcsname[1]{%
    \fcolorbox{#2}{#2}{\color{white}\textsf{\textbf{#1}}}
    {\color{#2}##1}}%
  \expandafter\newcommand\csname at#1\endcsname{%
    \fcolorbox{#2}{#2}{\color{white}\textsf{\textbf{@#1}}}
    {\color{#2}}}%
  \expandafter\newcommand\csname #1hl\endcsname[2]{%
    \colorbox{#2}{\color{white}\textsf{\textbf{#1}}}\sethlcolor{Azure2}\hl{##2}~%
    \expandafter\ifx\csname commentarrow\endcsname\relax$\leftarrow$\else \commentarrow[#2]\fi~%
    {\color{#2}##1}}%
  \expandafter\newcommand\csname #1st\endcsname[2]{%
    \colorbox{#2}{\color{white}\textsf{\textbf{#1}}}\sout{##2}~%
    \expandafter\ifx\csname commentarrow\endcsname\relax$\leftarrow$\else \commentarrow[#2]\fi~%
    {\color{#2}##1}}%
}
\newcommenter{TODO}{DodgerBlue1}
\newcommenter{rtron}{Green3}

\usepackage{comment}

\usepackage{pdfcomment}

\usepackage{soul}

\usepackage[normalem]{ulem}

\usepackage{csquotes}


%% file: preamble/graphicsTikz.tex
\usepackage{tikz}
\usetikzlibrary{calc}
\usetikzlibrary{matrix}
\usetikzlibrary{chains}
\usetikzlibrary{shapes.geometric}
\usetikzlibrary{arrows.meta}
\usetikzlibrary{decorations.pathreplacing}
\usetikzlibrary{backgrounds}

\tikzset{
  dim above/.style={to path={\pgfextra{
        \pgfinterruptpath
        \draw[>=latex,|->|] let
        \p1=($(\tikztostart)!1.5em!90:(\tikztotarget)$),
        \p2=($(\tikztotarget)!1.5em!-90:(\tikztostart)$)
        in(\p1) -- (\p2) node[pos=.5,sloped,above]{#1};
        \endpgfinterruptpath
      }
    }
  },
  dim double above/.style={to path={\pgfextra{
        \pgfinterruptpath
        \draw[>=latex,|->|] let
        \p1=($(\tikztostart)!3em!90:(\tikztotarget)$),
        \p2=($(\tikztotarget)!3em!-90:(\tikztostart)$)
        in(\p1) -- (\p2) node[pos=.5,sloped,above]{#1};
        \endpgfinterruptpath
      }
    }
  },
  dim below/.style={to path={\pgfextra{
        \pgfinterruptpath
        \draw[>=latex,|->|] let 
        \p1=($(\tikztostart)!-1em!-90:(\tikztotarget)$),
        \p2=($(\tikztotarget)!-1em!90:(\tikztostart)$)
        in (\p1) -- (\p2) node[pos=.5,sloped,below]{#1};
        \endpgfinterruptpath
      }
    }
  },
}

\tikzset{
    right angle quadrant/.code={
        \pgfmathsetmacro\quadranta{{1,1,-1,-1}[#1-1]}     
        \pgfmathsetmacro\quadrantb{{1,-1,-1,1}[#1-1]}},
    right angle quadrant=1, 
    right angle length/.code={\def\rightanglelength{#1}},   
    right angle length=2ex, 
    right angle symbol/.style n args={3}{
        insert path={
            let \p0 = ($(#1)!(#3)!(#2)$) in     
                let \p1 = ($(\p0)!\quadranta*\rightanglelength!(#3)$), 
                \p2 = ($(\p0)!\quadrantb*\rightanglelength!(#2)$) in 
                let \p3 = ($(\p1)+(\p2)-(\p0)$) in  
            (\p1) -- (\p3) -- (\p2)
        }
    }
}

\newcommand{\pgfextractangle}[3]{%
    \pgfmathanglebetweenpoints{\pgfpointanchor{#2}{center}}
                              {\pgfpointanchor{#3}{center}}
    \global\let#1\pgfmathresult  
}

\usetikzlibrary{shapes.arrows}
\newcommand{\commentarrow}[1][Azure4]{\tikz[baseline=-3pt]{\node[shape border uses incircle, fill=#1,rotate=180,single arrow, inner sep=1pt, minimum size=6pt, single arrow head extend=2pt]{};}}

%% file: preamble/robotics.tex
\tikzset{ax/.style={-latex,line width=2pt}}

\tikzset{camera/.style={fill=Sienna1,fill opacity=0.5},%
image plane/.style={draw=RoyalBlue3,line width=2pt}}
